\documentclass[%
reprint,
 amsmath,
 amssymb,
 aps,
 prd,
]{revtex4-2}
\usepackage{comment}
\usepackage{graphicx}% Include figure files
\usepackage{dcolumn}% Align table columns on decimal point
\usepackage{bm}% bold math
\usepackage{amsmath, amsfonts, amssymb,amsthm}
\usepackage{physics}
\usepackage{braket}
\usepackage{enumitem}
\usepackage[mode=buildnew]{standalone}
\usepackage{tikz}

\newtheorem{mydef}{Definition}

\newtheorem{theo}[mydef]{Theorem}
\newtheorem{lem}[mydef]{Lemma}
\usepackage{scalerel}
\usetikzlibrary{svg.path}

\definecolor{orcidlogocol}{HTML}{A6CE39}
\tikzset{
  orcidlogo/.pic={
    \fill[orcidlogocol] svg{M256,128c0,70.7-57.3,128-128,128C57.3,256,0,198.7,0,128C0,57.3,57.3,0,128,0C198.7,0,256,57.3,256,128z};
    \fill[white] svg{M86.3,186.2H70.9V79.1h15.4v48.4V186.2z}
                 svg{M108.9,79.1h41.6c39.6,0,57,28.3,57,53.6c0,27.5-21.5,53.6-56.8,53.6h-41.8V79.1z M124.3,172.4h24.5c34.9,0,42.9-26.5,42.9-39.7c0-21.5-13.7-39.7-43.7-39.7h-23.7V172.4z}
                 svg{M88.7,56.8c0,5.5-4.5,10.1-10.1,10.1c-5.6,0-10.1-4.6-10.1-10.1c0-5.6,4.5-10.1,10.1-10.1C84.2,46.7,88.7,51.3,88.7,56.8z};
  }
}

\newcommand\orcidicon[1]{\href{https://orcid.org/#1}{\mbox{\scalerel*{
\begin{tikzpicture}[yscale=-1,transform shape]
\pic{orcidlogo};
\end{tikzpicture}
}{|}}}}
\usepackage{hyperref}
\DeclareMathOperator{\Prob}{Prob}
\DeclareMathOperator{\expt}{E}
\DeclareMathOperator{\ch}{ch}

\begin{document}

%\preprint{APS/123-QED}

\title{Impossible measurements require impossible apparatus}

\author{Henning Bostelmann \orcidicon{0000-0002-0233-2928}}
 \email{henning.bostelmann@york.ac.uk}
 \author{Christopher J. Fewster \orcidicon{0000-0001-8915-5321}}
 \email{chris.fewster@york.ac.uk }
 \author{Maximilian H. Ruep \orcidicon{0000-0001-6866-4506}}
 \email{maximilian.ruep@york.ac.uk}
\affiliation{Department of Mathematics, University of York, Heslington, York YO10 5DD, United Kingdom
}

\date{\today}

\begin{abstract}
A well-recognised open conceptual problem in relativistic quantum field theory concerns the relation between measurement and causality. Naive generalisations of quantum measurement rules can allow for superluminal signalling (`impossible measurements'). This raises the problem of delineating physically allowed quantum measurements and  operations. We analyse this issue in a recently proposed framework in which local measurements (in possibly curved spacetime) are described physically by coupling the system to a probe. We show that the state-update rule in this setting is consistent with causality provided that the coupling between the system and probe is local. Thus, by establishing a well-defined framework for successive measurements, we also provide a class of physically allowed operations. Conversely, impossible measurements can only be performed using impossible (non-local) apparatus.
\end{abstract}

\maketitle

%\tableofcontents

\section{Introduction}
\label{sec:intro}
It is a central tenet of special and general relativity that there is a maximal speed of causal influence, the \emph{speed of light}: there can be no superluminal signalling. This should apply, in particular, to relativistic quantum field theory (QFT) and relativistic quantum information (RQI). However, as is very well known, the standard notion of \emph{measurement} challenges this tenet: it has been argued that `ideal measurements' in QFT can yield superluminal signalling~\cite{sorkin1993impossible, benincasa2014quantum} and that `nondemolition' measurements of Wilson loops in non-Abelian gauge theory can transfer charge over spacelike distances~\cite{beckman2002}. Even operations not directly associated to an ideal measurement, such as unitary transformations, can enable superluminal communication~\cite{benincasa2014quantum}. Those measurements that are in conflict with causality are called \emph{impossible measurements}~\cite{sorkin1993impossible} and their existence naturally raises the question of delineating (i) ``physically allowed quantum operations''~\cite{beckman2002}, as well as (ii) ``observables [that] can be measured consistently with causality''~\cite{sorkin1997forks}. These questions are not just of general conceptual importance~\cite{sorkin1997forks} but also directly affect applications in RQI~\cite{benincasa2014quantum, dowker2011useless} due to the lack of a clear-cut criterion for allowed operations that also allows an explicit construction of the latter. 

One way to address the difficulties just mentioned is to adopt an operational approach to measurement, in which the \emph{system} of interest is temporarily coupled to a measurement device (\emph{probe}); following a measurement of a probe observable the probe is discarded (traced out). This constitutes a \emph{measurement scheme}~\cite{busch2016quantum} for  
an induced observable of the system and, importantly, yields an associated state-update rule. Although well established in quantum mechanics, this idea was only recently adapted to QFT in possibly curved spacetimes, thus implementing the concept of a measurement scheme in a local and covariant way~\cite{fewster2018quantum} (see~\cite{fewster2019generally} for a summary). We call this the FV framework (after the authors of~\cite{fewster2018quantum}) and its elements FV measurement schemes.

In this paper we show that, due to the locality of the coupling between \emph{system} and \emph{probe}, measurements in the FV framework are \emph{not} plagued by superluminal signalling in the sense of~\cite{sorkin1993impossible}; i.e., an impossible measurement requires an impossible, non-local apparatus. As a result, FV measurement schemes are consistent with causality and the associated state-updates provide a large and explicitly calculable class of ``physically allowed quantum operations''; a significant improvement to the ``case-by-case analysis''~\cite{sorkin1997forks}, which has been the approach of all previous literature on this topic to the best of our knowledge. This reinforces the usefulness of the FV framework for treating measurements in QFT, pointing to its use as a general way to understand which operations are physically allowable, and hence forming useful underpinning for applications in RQI.

The structure of our paper is as follows: in Sec.~\ref{sec:heuristic_overview} we recall the essence of the causality problem posed by Sorkin~\cite{sorkin1993impossible} in the form of a tripartite signalling protocol applicable in flat as well as curved spacetime; this is followed by a non-technical motivation and discussion of FV measurement schemes and a statement of our main result. Section~\ref{sec:technical} comprises the precise presentation of the FV framework in the language of algebraic quantum field theory and in Sec.~\ref{sec:main_result} we show that, in the FV framework, Sorkin's protocol does not result in any acausal effects. Section~\ref{sec:causal_fac} consists of a discussion of causal factorisation and forms the basis of our discussion of multiple observers covered in Sec.~\ref{sec:multiple_obs}. In particular, we demonstrate how the state updates associated to selective and non-selective measurements (postulated in Secs.~\ref{sec:heuristic_overview} and~\ref{sec:main_result}) can be derived from the principle of causal factorisation. The consequences on causality of this analysis are displayed in Sec.~\ref{sec:Sorkin-N}, where we explicitly show that the FV framework consistently describes any finite number of causally orderable measurements without any superluminal signalling issues. As a last point we conclude and provide an outlook in Sec.~\ref{sec:conclusion}.

\section{Heuristic overview}
\label{sec:heuristic_overview}

\begin{figure}[ht]
\includegraphics[width=.4\textwidth]{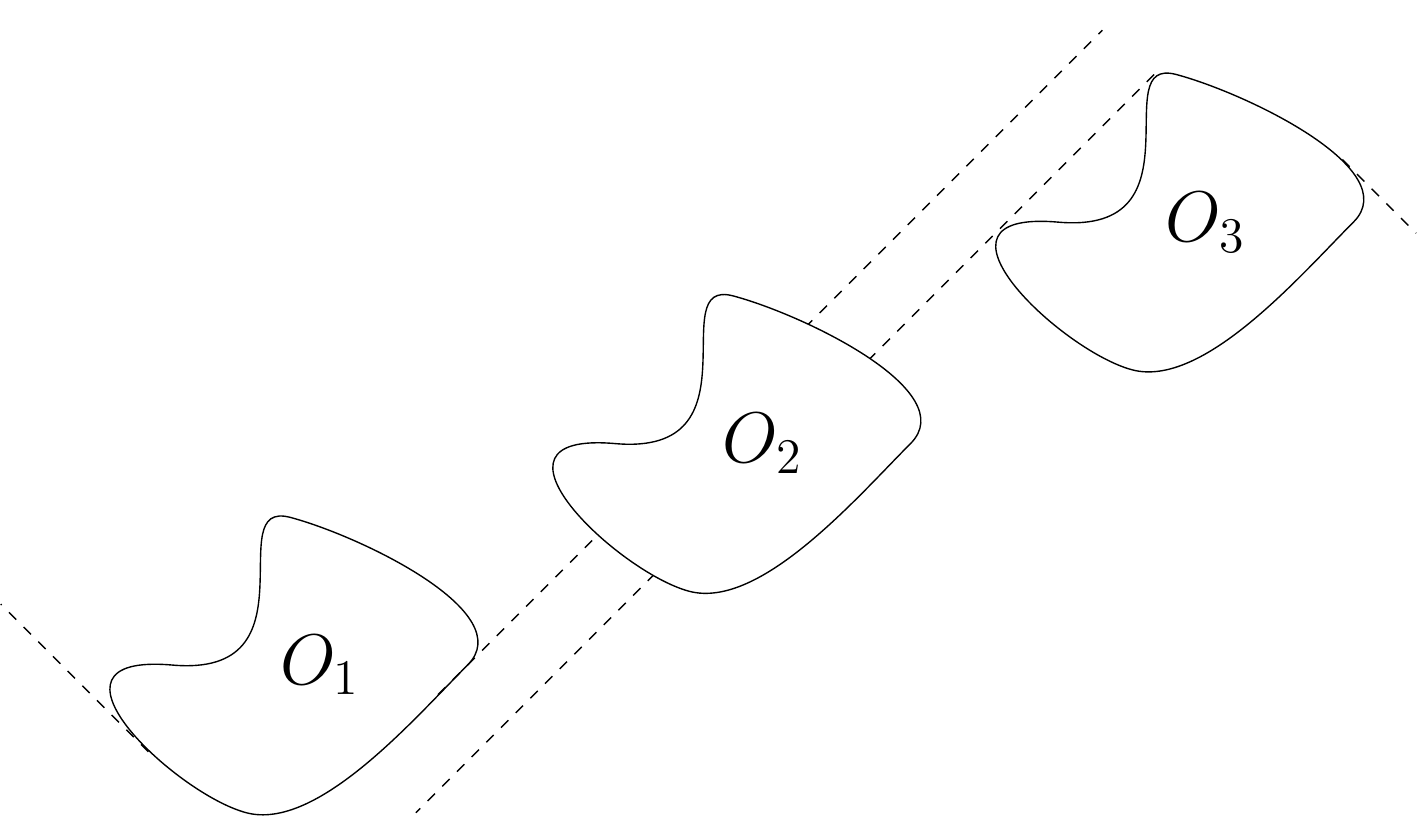}
\caption{\label{fig_1} Schematic spacetime diagram of the relative causal position of the regions $O_1$, $O_2$ and $O_3$.}
\end{figure}

\subsection{Superluminal signalling \`a la Sorkin} 

Sorkin has argued~\cite{sorkin1993impossible} that the notion of an ideal measurement conflicts with locality and causality when extended from quantum mechanics to QFT. In particular, he presented the following protocol: let Alice, Bob and Charlie be three experimenters in three laboratories performing actions in the spacetime `regions of control' $O_1,O_2,O_3$ such that parts of $O_1$ are in the past of $O_2$ and parts of $O_2$ are in the past of $O_3$, but such that $O_1$ is spacelike separated from $O_3$ as shown in Fig \ref{fig_1}. Let $A$ be a local observable of $O_1$, e.g., an algebraic combination of quantum fields smeared against test functions vanishing outside $O_1$. Define $B, C$ similarly and let $\rho$ be the initial state of the quantum field. Sorkin considers the following tripartite procedure. In step one, Alice performs a local measurement of $A$ in her laboratory. In the absence of any post-selection in the experimental data analysis, the resulting updated state is a probabilistic mixture, i.e., a convex combination of states, each selected on a different possible outcome of the measurement, weighted by the respective probabilities. This updated state is denoted $\rho_A$. In step two, Bob measures $B$, producing a further (similar) update $\rho_A \mapsto \rho_{AB}$. In step three, Charlie measures observable $C$ in state $\rho_{AB}$. Since Charlie's laboratory is spacelike separated from Alice's, $\Tr (\rho_{AB} \; C)$ should (in the absence of superluminal communication) give the same result as $\Tr (\rho_{B} \; C)$ - the situation where Alice does not measure at all. This condition, Sorkin argues, puts non-trivial constraints on feasible (ideal) measurements, to the extent that ``it becomes a priori unclear, for quantum field theory, which observables can be
measured consistently with causality and which can't. This would seem to deprive [QFT] of any definite measurement theory,
leaving the issue of what can actually be measured to (at best) a case-by-case analysis''~\cite{sorkin1997forks}. By contrast, we will show that the FV framework furnishes QFT with a definite measurement theory.

\subsection{The idea behind the FV framework} 

A measurement scheme in quantum measurement theory is the theoretical description of a measurement on a \emph{system}, prepared in state $\rho_S$, by the operational procedure of bringing it into contact with a \emph{probe}, itself to be regarded as a quantum system, and initially prepared in state $\rho_P$. The `contact' between system and probe is modelled by coupling them together via interactions. In quantum mechanics, this is achieved by an interacting unitary time-evolution which operates for a short period of time and is then removed. A subsequent measurement made on the probe is interpreted as a measurement of the system, and indeed it is possible to establish a correspondence between observables of the probe and induced observables of the system. One says that the combination of the probe, interacting dynamics and probe observable form a \emph{measurement scheme} for the induced system observable (see~\cite{busch2016quantum} for a comprehensive account). 

The FV framework translates the above idea to QFT in possibly curved spacetime; equally, it can incorporate QFT under the influence of external fields. It is phrased in terms of the algebraic approach to QFT~\cite{haag2012local} (see~\cite{fewster2019algebraic} for an introduction), but for the purposes of the following discussion we use  familiar terminology of QFT; the more formal algebraic version will be set out in Sec.~\ref{sec:technical} and used in our proof.

We consider two local relativistic QFTs, modelling the system and the probe. Taking a tensor product, they may be combined as a single theory with no coupling between them. If the two theories are obtained from Lagrangian densities $\mathcal{L}_S$ and $\mathcal{L}_P$, the uncoupled combination is defined by the sum $\mathcal{L}_S+\mathcal{L}_P$.  
The contact between system and probe is modelled by another QFT, in which the two are coupled so that the coupling is only effective within a compact set $K$ of spacetime, the \emph{coupling zone}. Crucially, it is assumed that this coupled QFT is itself a local relativistic theory. For Lagrangian theories, the coupled theory would be described by a \emph{local} coupling term such as $\mathcal{L}_I:=- \lambda \alpha(x) \phi(x) \psi(x)$,
where $\phi$ and $\psi$ are system and probe Hermitian scalar fields, respectively, and the real-valued smooth function $\alpha$, perhaps representing an external field, vanishes outside $K$. However, we emphasise that our results are not tied to this particular coupling $\mathcal{L}_I$, nor is it even required that the theories involved are described by Lagrangians. 

These assumptions allow for a direct identification between the free theory and the interacting theory \emph{before} as well as \emph{after} the coupling -- or more precisely, outside $K$'s causal future and past, respectively.
The comparison between these identifications is encoded in a unitary scattering matrix $\sf S$, which takes the place of the interacting time evolution in the quantum mechanical setting. To be specific, the adjoint action $A\mapsto {\sf S}A{\sf S}^{-1}$ of $\sf S$ is obtained by mapping from the uncoupled to coupled theory using the late-time identification, followed by mapping back to the uncoupled system using the early-time identification. This corresponds to the usual composition of M\o ller maps. (Assuming the coupled and uncoupled theories both have the \emph{time-slice property} -- see below -- any observable can be expressed in terms of either late-time or early-time observables.) In the example above the scattering map is given by ${\sf S}:=\overline{\mathcal{T}} \exp{\frac{\mathrm{i} \lambda}{\hbar} \int_K \alpha(x) \phi(x) \psi(x) \mathrm{d}x}$ to all orders in formal perturbation theory, where $\overline{\mathcal{T}} \exp$ is the \emph{anti}-time-ordered exponential and $\alpha$ functions as a smooth spacetime cutoff.

The locality of the theories under consideration is reflected in localisation properties of $\sf S$, which are discussed in more detail in Lemma \ref{lem_theta_spacelike_loc_change} below. In consequence, the idea of a measurement scheme can be implemented in QFT as a \emph{local} concept. In particular, it was shown in~\cite{fewster2018quantum,fewster2019generally} how the correspondence between probe observables and induced system observables may be made, and how rules for state update appropriate to selective and non-selective measurements may be described. 
A non-technical outline of these results now follows.

Suppose that $Z$ is a local observable of the probe theory, corresponding to a local observable $\openone\otimes Z$ of the uncoupled combination of the system and probe theories. Likewise, $\rho_S\otimes\rho_P$ is an uncorrelated state of the same theory. We consider an experiment in the coupled theory, in which an observable corresponding to $\openone\otimes Z$ at late times is measured in a state that corresponds to $\rho_S\otimes\rho_P$ at early times. The expectation value of this measurement,
$\expt_{\rho_P}(Z;\rho_S)$, is~\cite{fewster2018quantum}
\begin{equation}
    \expt_{\rho_P}(Z;\rho_S) = 
    \Tr \qty((\rho_S \otimes \rho_P) ({\sf S}(\openone \otimes Z) {\sf S^\dagger}) ).
    \label{eq_non_technical_single_exp_value}
\end{equation}
We remark that we use the term \emph{expectation value} in an operational way following the frequentist interpretation. For an analysis of the interpretation in terms of the first moment of an underlying probability measure we refer the reader to~\cite{Drago_2020}. The induced system observable $\widehat{Z}_{\rho_P}$ corresponding to probe observable $Z$ is, by definition, the observable whose expectation in state $\rho_S$ matches that of the actual experiment:
\begin{equation}
  \Tr \qty( \rho_S  \widehat{Z}_{\rho_P}) =\expt_{\rho_P}(Z;\rho_S) 
  \label{eq_non-technical_induced_observable}
\end{equation}
It turns out that the induced observable belongs to the algebra of system observables corresponding to the coupling region (more precisely, to the causal hull of any connected neighbourhood thereof).

Turning to the issue of state updates, let us, for the sake of presentation, consider the case where there is a measurement of a probe observable $Z$ and a (not further specified) measurement of a system observable $A$ in the `out' region (i.e., not to the past) of the $Z$-measurement. (A more formal discussion solely in terms of probe observables yields the same expressions for the updated states and is given in Sec.~\ref{sec:multiple_obs}.) For simplicity let us assume that $Z$ is a yes-no observable (i.e., an effect).
Analysing the $A$ and $Z$ measurement results together over an ensemble of identical runs, we may restrict to the subensemble in which the $Z$-measurement was successful (`yes'). The expectation value of $A$, conditioned on success of $Z$, is  
\begin{equation}
  \frac{\Tr \qty((\rho_S \otimes \rho_P) ({\sf S}(A \otimes Z) {\sf S^\dagger}) ) }{\Prob(Z|\rho_S)} =: \Tr\qty(\rho'_{S|Z} A),
  \label{eq_non_technical_conditional_expectation}
\end{equation}
where the system state $\rho'_{S|Z}$ defined in this way may be regarded as the updated state consequent upon successful measurement of $Z$,
which occurs in the full ensemble with probability $\Prob_{\rho_P}(Z|\rho_S)=\expt_{\rho_P}(Z;\rho_S)$. In a similar way, the updated state $\rho'_{S|\neg Z}$ conditioned on an unsuccessful (`no') $Z$-measurement 
may be obtained from the above on replacing $Z$ by $\openone-Z$. If no selection is made, then the updated state $\rho'_S$ is an appropriately weighted statistical mixture of $\rho'_{S|Z}$ and $\rho'_{S|\neg Z}$, giving
\begin{equation}
    \Tr \qty(\rho_S' C) = \Tr \qty((\rho_S \otimes \rho_P) ({\sf S}(C \otimes \openone) {\sf S^\dagger}) ).
    \label{eq_qm_single_state_update}
\end{equation}
Notice that this expression is independent of the particular observable $Z$; in tracing out the probe degrees of freedom, it is assumed that no further measurements of the probe are made.

Multiple measurements, each conducted by a different probe, may be accommodated provided that their coupling regions lie in a causal order, with each separated by a Cauchy surface from its predecessor. A crucial consistency relation established in~\cite{fewster2018quantum} implies that the rules for state updates are independent of the choice of order when more than one is possible; this was shown explicitly in~\cite{fewster2018quantum} for pairs of measurements and will be extended to the general case in Sec.~\ref{sec:multiple_obs} below. The consistency result relies on a natural assumption called \emph{causal factorisation}.

\subsection{Sorkin's protocol in the FV framework}

In the FV framework, Sorkin's protocol is modelled as follows: Alice, Bob and Charlie are each described by \emph{probes} which are coupled to the \emph{system} of interest in the compact coupling zones $K_1, K_2, K_3$ each contained in the connected \emph{regions} $O_1, O_2, O_3$ respectively, in which the experimenters perform actions. This guarantees that there is a natural causal order of $K_1, K_2, K_3$, i.e., the one inherited from $O_1, O_2, O_3$, see Fig \ref{fig_1}. In particular, there are Cauchy surfaces having $K_1$ to their past and $K_2$ to their future;  
$K_2$ and $K_3$ are also separated by Cauchy surfaces in the same way.

The measurements of Alice and Bob in step one and two of the protocol produce an update of the system state $\rho_S \mapsto \rho_{AB}$ according to
\begin{equation}
    \begin{aligned}
    &\Tr \qty(\rho_{AB} C)\\
    &= \Tr \qty((\rho_S \otimes \rho_{P_1} \otimes \rho_{P_2}) ( { \sf S_1  S_2}(C \otimes \openone \otimes \openone) { \sf S_2^\dagger S_1^\dagger}) ),
    \label{eq_qm_AB_state_update}
    \end{aligned}
\end{equation}
which is a straightforward generalisation of Eq.~\eqref{eq_qm_single_state_update} using the natural causal order of the three experimenters. In fact there is no ambiguity if $K_1$, $K_2$ and $K_3$ admit other causal orders (which can happen if $K_2$ is spacelike from $K_1$ or $K_3$) -- see Sec.~\ref{sec:causal_fac}.
As argued above, the expectation value of Charlie's measurement in step three is given by Eq.~\eqref{eq_qm_AB_state_update} for a probe-induced system observable $C$, which is determined by the interaction between Charlie's probe and the system in coupling zone $K_3$ and may be localised in $O_3$. The superluminal signalling between the spacelike separated experimenters Alice and Charlie in Sorkin's protocol arises if $\Tr \qty(\rho_{AB} C)$ differs from $\Tr \qty(\rho_{B} C)$, where $\rho_B$ is the updated state in a situation where Alice does not perform an experiment; i.e., where there is \emph{no} coupling between her probe and the system and hence no measurement is made on the system. This corresponds to Eq.~\eqref{eq_qm_AB_state_update} in the case where ${\sf S_1} = \openone$. Hence, there is \emph{no} superluminal signalling if 
\begin{equation}
    \begin{aligned}
    &\Tr \qty((\rho_S \otimes \rho_{P_1} \otimes \rho_{P_2}) ({\sf S_1 S_2}(C \otimes \openone \otimes \openone) { \sf S_2^\dagger S_1^\dagger}))\\
    &=\Tr \qty((\rho_S\otimes \rho_{P_2}) ({\sf S_2}(C \otimes \openone) {\sf S_2^\dagger}))
    \end{aligned}
    \label{eq_no_signalling}
\end{equation}
for system observables $C$ induced by Charlie's probe. 
The main result of this paper is that~\eqref{eq_no_signalling} holds under very mild technical assumptions. This result is stated and proved as Theorem \ref{theo_no_signalling} in Sec.~\ref{sec:main_result} and makes essential use of the localisation properties of the scattering map.

In fact, the statement we prove is actually more general as it establishes the desired equality for \emph{all} system observables $C$ localisable in $O_3$ and not just the ones induced by Charlie's probe (if this class is smaller).

\section{Technical description}\label{sec:technical}

The setting of the FV framework is algebraic quantum field theory in possibly curved spacetime, which we now briefly recall.

\subsection{Lorentzian geometry}

We start by fixing notation and recalling standard results of Lorentzian geometry. Let $M$ be a globally hyperbolic spacetime, i.e., a time-oriented Lorentzian spacetime of dimension at least two that contains a Cauchy surface. For $N \subseteq M$ let $J^+(N)$ and $J^-(N)$ denote its causal future and past, respectively, and define its causal hull to be $\ch (N):=J^+(N) \cap J^-(N)$; $N$ is called causally convex if it equals its causal hull. 
Any open causally convex subset of $M$ will be called a \emph{region} and is
itself globally hyperbolic when regarded as a spacetime in its own right. Let $D^+(N)$ and $D^-(N)$ denote the future and past Cauchy developments of $N$, that is, the set of points $p \in M$ such that every past-, respectively, future-inextendible piecewise smooth causal curve through $p$ intersects $N$. Then $D(N):= D^+(N) \cup D^-(N)$ is called the Cauchy development or domain of dependence of $N$. The causal complement of a subset $K$ is defined to be $K^\perp := M \setminus (J^+(K) \cup J^-(K))$. For a compact subset $K$, the sets $M \setminus J^\mp(K)$ and $K^\perp$ are all open and causally convex and therefore globally hyperbolic. See, for example, the Appendix of~\cite{Fewster_2012} for details and proofs.

\subsection{Algebraic quantum field theory}

Let $M$ be a globally hyperbolic spacetime. An algebraic quantum field theory (AQFT), or simply a \emph{theory}, on $M$ consists of a $*$-algebra $\mathcal{A}$ with a unit $\openone$, together with a family of sub-$*$-algebras $\mathcal{A}(N)$ of $\mathcal{A}(M):=\mathcal{A}$, each containing $\openone$ and labelled by the regions $N\subseteq M$. The elements of $\mathcal{A}(N)$ are considered to be local observables of $N$, e.g., algebraic combinations of smeared fields `$\int_N f(x) \phi(x) \; \mathrm{d}x$' for a quantum field $\phi$ and a test function $f$ vanishing outside $N$.
This interpretation motivates the following additional assumptions: 

\paragraph*{Isotony:} for regions $N_1 \subseteq N_2$: $\mathcal{A}(N_1) \subseteq \mathcal{A}(N_2)$.

\paragraph*{Einstein causality:} for spacelike separated regions $N_1$ and $N_2$: the elements of $\mathcal{A}(N_1)$ commute with the elements of $\mathcal{A}(N_2)$. 

\paragraph*{Time-slice property:} for regions $N_1 \subseteq N_2$, so that $N_1$ contains a Cauchy surface for $N_2$: $\mathcal{A}(N_1) = \mathcal{A}(N_2)$. 

The time-slice property encodes the existence of a (not further specified) local dynamical law. Morally: a quantum field is determined by its data on a Cauchy surface. We emphasise that the time-slice property is local in the sense that it applies to every region $N_2$.

Due to time-slice (and isotony), every observable is localisable in many different, possibly disjoint regions. For example, if an observable $A$ is localisable in a region $N_1$ and $N_2$ is a disjoint region containing $N_1$ in its domain of dependence, i.e., $N_1 \subseteq D(N_2)$, then $A$ is also localisable in $N_2$.

One also assumes a \emph{Haag property}, which heuristically guarantees that the theory captures all relevant degrees of freedom. It is used to show that induced observables are localisable in every \emph{connected} region containing the coupling zone -- see~\cite{fewster2018quantum} for details.

\paragraph*{Haag property:} for every compact set $K \subseteq M$ and every \emph{connected} region $L$ containing $K$, $\mathcal{A}(L)$ contains every $C\in \mathcal{A}$ that commutes with all elements of $\mathcal{A}(K^\perp)$.~\footnote{\label{footnote_Haag}The stated Haag property could be strengthened to incorporate all (possibly disconnected) regions $L$ containing $K$. However it might then conflict with a further desirable property called \emph{additivity}, see p.\ 146 in~\cite{haag2012local}.}

In AQFT, a \emph{state} is a linear map $\omega:\mathcal{A}(M)\to \mathbb{C}$ which assigns expectation values to algebra elements and is therefore required to be normalised, $\omega(\openone)=1$, and positive, $\omega(A^*A)\ge 0$ for all $A\in\mathcal{A}(M)$.

\subsection{Coupled theories}

The coupling between probe and system theory and the resulting scattering map arise as follows: suppose we have three theories on a globally hyperbolic spacetime $M$: a system-theory $\mathcal{S}$, a probe-theory $\mathcal{P}$ and a coupled theory $\mathcal{C}$, which mirrors the crucial assumption that the coupled structure is itself \emph{local}. Let $\mathcal{S} \otimes \mathcal{P}$ denote the tensor-product theory, i.e., the uncoupled combination. As discussed before, $\mathcal{S}$ and $\mathcal{P}$ are coupled together \emph{only} in a compact coupling zone $K \subseteq M$, which is modelled by the existence of a bijective, structure and localisation preserving identification between the coupled and uncoupled theories \emph{outside} (the causal hull) of $K$, see~\cite{fewster2018quantum} for the details. For the \emph{in-region} $M^-$ and \emph{out-region} $M^+$ defined by $M^\pm:=M \setminus J^\mp(K)$, this gives us the following maps:
\begin{equation}
    \begin{aligned}
    &\mathcal{S} \otimes \mathcal{P} \to \big(\mathcal{S} \otimes \mathcal{P}\big)(M^+) \to \mathcal{C}(M^+) \to \mathcal{C},\\
    &\mathcal{C} \to \mathcal{C}(M^-) \to \big(\mathcal{S} \otimes \mathcal{P}\big)(M^-) \to \mathcal{S} \otimes \mathcal{P},
    \end{aligned}
    \label{eq_arrows}
\end{equation}
each of which is an isomorphism. The first, third, fourth and sixth are given by the time-slice property as $M^\pm$ each contain a Cauchy surface for $M$~\cite{Fewster_2012}. The other arrows are given by the localisation preserving identification map. The overall composition defines the \emph{scattering map} $\Theta: \mathcal{S} \otimes \mathcal{P} \to \mathcal{S} \otimes \mathcal{P}$, which is an automorphism preserving algebraic relations but \emph{not} localisation.
Our earlier discussion in Sec.~\ref{sec:heuristic_overview} implicitly assumed that $\Theta$ was implemented as the adjoint action of a unitary scattering operator $\sf S$, i.e., $\Theta(A)={\sf S}A{\sf S^\dagger}$, but this is neither needed nor assumed in what follows. The localisation properties of $\Theta$ are summarised in the following lemma. 

\begin{lem}[Proposition 3.1(b),(c) in~\cite{fewster2018quantum}] \hspace{1mm}
\begin{enumerate}[leftmargin=0.5cm]
    \item For every region $N \subseteq K^\perp: \Theta$ acts trivially on \linebreak $\big(\mathcal{S} \otimes \mathcal{P}\big)(N)$.
    \label{lem_theta_spacelike_loc_change_property_trivial_action_spacelike_separation}
    \item For every region $N \subseteq M^+$ \hspace{-2mm} and every region $N^- \subseteq M^-$ with $N \subseteq D(N^-): \; \Theta \big(\mathcal{S} \otimes \mathcal{P}\big)(N) \subseteq \big(\mathcal{S} \otimes \mathcal{P}\big)(N^-)$. 
\end{enumerate}
\label{lem_theta_spacelike_loc_change}
\end{lem}
The first property captures the idea that the coupling has no effect in spacelike separated regions, whereas the second property indicates how $\Theta$ changes the localisation of observables.

Now suppose that the system is prepared in state $\omega$ and the probe in state $\sigma$, and that a measurement of a probe observable is made. The state update rule (without selection) is that $\omega\mapsto \omega'$, where
\begin{equation}\label{eq_fv_single_state_update}
    \omega'(C) = (\omega\otimes\sigma)(\Theta(C\otimes\openone)),
\end{equation}
which is readily recognised as the analogue of~\eqref{eq_qm_single_state_update}.

\section{Main result}\label{sec:main_result}

Let us now discuss the rigorous FV version of Sorkin's protocol and Eq.~\eqref{eq_qm_AB_state_update}. Alice, Bob and Charlie each perform actions in the connected \emph{regions} $O_1, O_2, O_3$. We assume they fulfill (i) $O_2 \cap J^-(O_1) = \emptyset$; (ii) $O_3 \cap J^-(O_2) = \emptyset$; (iii) $O_3$ is spacelike separated form $O_1$; (iv) $O_3$ has 
compact closure $\overline{O}_3$. Note that this covers the situation sketched in Fig \ref{fig_1} but is more general. Let $\mathcal{S}$ be the system theory and let $\mathcal{P}_1,\mathcal{P}_2$ be the two probe theories of Alice and Bob with compact coupling zones $K_1,K_2$ contained in the regions $O_1, O_2$, respectively. Denote the corresponding in- and out-regions by $M_1^\mp,M_2^\mp$, the initial states by $\sigma_1,\sigma_2$ and the associated scattering maps by $\Theta_i: \mathcal{S} \otimes \mathcal{P}_i \to \mathcal{S} \otimes \mathcal{P}_i$ for $i=1,2$. On $\mathcal{S} \otimes \mathcal{P}_1 \otimes \mathcal{P}_2$ define $\hat{\Theta}_1 := \Theta_1 \otimes_3 \openone$ and $\hat{\Theta}_2 := \Theta_2 \otimes_2 \openone$, where the subscript on the tensor product indicates the slot into which the second factor is inserted. 

Let $C$ be a system observable localisable in $O_3$, Charlie's `region of control'. For example, $C$ could be
the induced observable corresponding to any probe observable of Charlie's. Owing to assumptions (i--iii), Alice, Bob and Charlie admit a causal order in which Alice's region precedes Bob's, and Bob's region precedes Charlie's. If Alice and Bob each perform a measurement, the expectation value for Charlie's measurement is therefore given by
\begin{equation}
    \begin{aligned}
    \omega_{AB}(C):= (\omega \otimes \sigma_1 \otimes \sigma_2)((\hat{\Theta}_1 \circ \hat{\Theta}_2)(C \otimes \openone \otimes \openone)),
    \end{aligned}
    \label{eq_successive_update}
\end{equation}
for initial system state $\omega$. Strictly speaking, when writing down Eq.~\eqref{eq_successive_update} at this stage, we make the assumption that the effect of two causally orderable measurements on the initial state $\omega$ is given in terms of a composition of individual state updates. However, further below in Sec.~\ref{sec:multiple_obs} we show how Eq.~\eqref{eq_successive_update} can be \emph{derived} in the FV framework. Furthermore, assumptions (i--iii) do not exclude the possibility that the regions controlled by Alice, Bob and Charlie also admit other causal orderings, but Charlie's expectation value is well-defined and independent of any choices made. This will also be discussed in greater depth in Sec.~\ref{sec:causal_fac}. On the other hand, if Alice does not perform her experiment, Charlie's expectation value is
\begin{equation}
    \omega_B(C) = (\omega\otimes\sigma_2)(\Theta_2 (C\otimes \openone)).
\end{equation}

 The following theorem [the rigorous analogue of Eq.~\eqref{eq_no_signalling}] shows that Sorkin's protocol does \emph{not} signal in the FV framework. Note that it gives the desired equality without the (possibly restricting) assumption that $C$ is an induced observable.
\begin{theo}
    \label{theo_no_signalling}
In the notation above, suppose the following assumptions hold: (a) $K_2 \cap J^-(K_1) = \emptyset$; (b) $O_3$ is a region with compact closure; (c) $O_3 \cap J^-(K_2) =\emptyset$; (d) $\overline{O}_3$ is spacelike separated from $K_1$. Then
 \begin{equation}
        \begin{aligned}
        \forall C \in \mathcal{S}(O_3): (\hat{\Theta}_1 \circ \hat{\Theta}_2) (C \otimes \text{\rm $\openone$} \otimes \text{\rm $\openone$}) = \hat{\Theta}_2 (C \otimes \text{\rm $\openone$} \otimes \text{\rm $\openone$}).
        \end{aligned}
    \end{equation}
\end{theo}
\noindent
This immediately implies
\begin{equation}
    \begin{aligned}
         \omega_{AB}(C) &= (\omega \otimes \sigma_1 \otimes \sigma_2) ((\hat{\Theta}_1 \circ \hat{\Theta}_2)(C \otimes \openone \otimes \openone))\\
         &= (\omega \otimes \sigma_2)(\Theta_2(C \otimes \openone)) = \omega_B(C),
    \end{aligned}
\end{equation}
i.e., Charlie's measurement outcome is independent of whether Alice does or does not perform an experiment at all. There is no superluminal signalling.\\

The proof of Theorem~\ref{theo_no_signalling} relies on localisation properties of the scattering map, combined with a geometrical lemma.

\begin{figure}[ht]
\includegraphics[width=.4\textwidth]{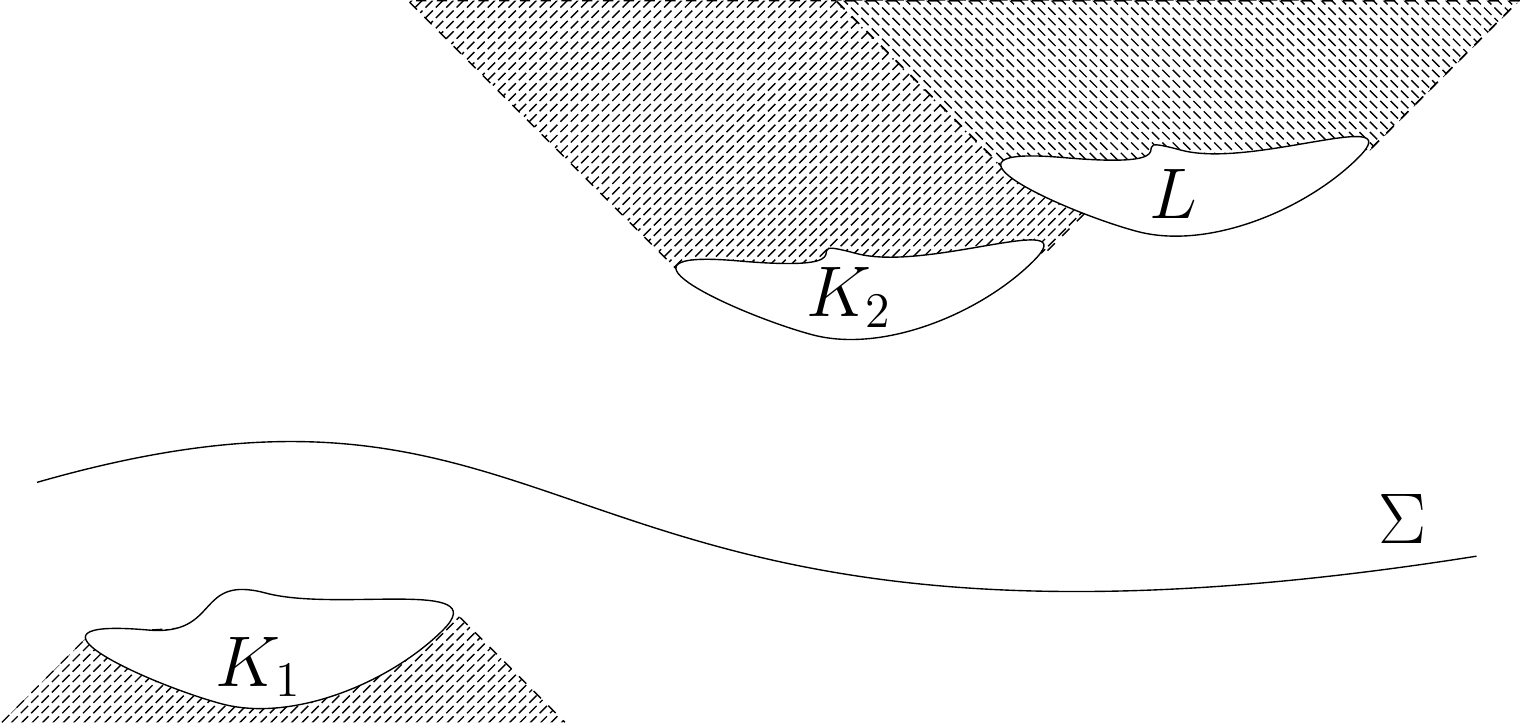}
\caption{\label{fig_2} Schematic spacetime diagram of the relative causal position of the compact sets $K_1$, $K_2$ (coupling zones) and $L$ ($\overline{O}_3$ in Lemma \ref{lem_existence_L-}) as well as the Cauchy surface $\Sigma$ in Lemma \ref{lem_Cauchy_surface}.}
\end{figure}

\begin{lem}
Let $K_1, K_2, L$ be compact subsets of $M$, let $K_2 \cap J^-(K_1) = \emptyset$ and $L \cap J^-(K_1) = \emptyset$. Then there exists a Cauchy surface $\Sigma$ of $M_1^+$ such that $\Sigma \subseteq M \setminus (J^-(K_1) \cup J^+(K_2) \cup J^+(L))$, as sketched in Fig.~\ref{fig_2}.
\label{lem_Cauchy_surface}
\end{lem}

\begin{proof}
$M_1^+ = M \setminus J^-(K_1)$ is globally hyperbolic (see Lemma A.4 in~\cite{Fewster_2012}). By
Proposition 4 in~\cite{bernal2003smooth} (due to Geroch~\cite{geroch1970domain}) there exists a surjective, continuous function $t:M_1^+ \to \mathbb{R}$, strictly increasing on every future-directed causal curve, whose level sets are Cauchy surfaces for $M_1^+$. Since $K_2$ and $L$ are compact and $t$ is continuous, $\tilde{\tau}:= \min{t[K_2 \cup L]}$ exists. Choose $\tau < \tilde{\tau}$ and set $\Sigma:= t^{-1}[\{\tau\}]$. $\Sigma$ is a Cauchy surface for $M_1^+$ and fulfills the desired properties.  
\end{proof}

We apply the lemma for the case where $K_1,K_2$ are the coupling zones of Alice and Bob and $L$ is the closure of Charlie's  region of control, i.e., $L = \overline{O}_3$, which is compact. This allows us to prove that $O_3$ is contained in the domain of dependence of $K_1^\perp \cap M_2^-$. 

\begin{lem}
Let $K_1,K_2$ be compact subsets of $M$ such that $K_2 \cap J^-(K_1) = \emptyset$. Then for every  region $O_3$ with compact closure such that $O_3 \cap J^-(K_2) = \emptyset$ and $\overline{O}_3 \subseteq  K_1^\perp$ it holds that $O_3 \subseteq D(K_1^\perp \cap M_2^-)$.
\label{lem_existence_L-}
\end{lem}

\begin{proof}
By setting $L:=\overline{O}_3$ and using Lemma \ref{lem_Cauchy_surface}, we can find $\Sigma$, a Cauchy surface for $M_1^+$ which lies in $M \setminus (J^-(K_1) \cup J^+(K_2) \cup J^+(\overline{O}_3))$. Set $T:= J^-(\overline{O}_3) \cap \Sigma \subseteq K_1^\perp \cap M_2^-$. [$T$ is spacelike separated from $K_1$, because $\Sigma$ is disjoint from $J^-(K_1)$ and because $J^-(\overline{O}_3)$ is disjoint from $J^+(K_1)$ as $\overline{O}_3 \subseteq K_1^\perp$ by assumption.] Now $O_3 \subseteq D(T)$, while $K_1^\perp \cap M_2^- = M_1^+ \cap M_1^- \cap M_2^-$; as $M_1^+, M_1^-$ and $M_2^-$ are open and causally convex (see Lemma A.4 in~\cite{Fewster_2012}), so is their intersection, i.e., it is a region, and since it contains $T$, we have that $O_3 \subseteq D(T) \subseteq D(K_1^\perp \cap M_2^-)$.
\end{proof}

Theorem \ref{theo_no_signalling} now follows by using the localisation properties of $\Theta$ and the fact that Charlie's region of control is contained in the domain of dependence of a sub-region of $K_1^\perp$. 

\begin{proof}[Proof of Theorem \ref{theo_no_signalling}]
Since $C \in \mathcal{S}(O_3)$, $C \otimes \openone \otimes \openone$ can be localised in $O_3$ too. According to Lemma \ref{lem_existence_L-}, $O_3 \subseteq  D(K_1^\perp \cap M^-_2)$. According to Lemma \ref{lem_theta_spacelike_loc_change}, we know that $\hat{\Theta}_2(C \otimes \openone \otimes \openone)$ can be localised in the region $K_1^\perp \cap M^-_2$. But since $K_1^\perp \cap M^-_2 \subseteq K_1^\perp$, we have by Lemma \ref{lem_theta_spacelike_loc_change} that $\qty(\hat{\Theta}_1 \circ \hat{\Theta}_2) \qty(C \otimes \openone \otimes \openone) = \hat{\Theta}_2 \qty(C \otimes \openone \otimes \openone)$.
\end{proof}

\section{Causal factorisation}\label{sec:causal_fac}

In the previous section we showed that measurements of three observers described in the FV framework do not run into the potential superluminal signalling issues associated to Sorkin's \emph{impossible measurements}. To do this, we made the assumption that the effect of causally orderable measurements may be given in terms of a composition of individual state updates as in Eq.~\eqref{eq_successive_update}. In the next section we will show that this assumption can actually be derived in the FV framework as a result of what is called \emph{causal factorisation}, which we now describe.  Our presentation here is certainly not the most general possible but will be sufficient for our current purposes. We intend to report elsewhere on more abstract and general properties of causal factorisation.

To start, let $\mathcal{K}$ be a collection of
compact spacetime subsets. A linear order $\le$ on $\mathcal{K}$ is said to be a \emph{causal linear order} if $K< K'$ implies $J^-(K)\cap J^+(K')=\emptyset$ for every $K,K'\in \mathcal{K}$. It follows that whenever $K<K'$, there is a Cauchy surface of $M$ with $K$ to its past and $K'$ to its future. If $\mathcal{K}$ admits a causal linear order, we say that $\mathcal{K}$ is \emph{causally orderable}. A causally orderable set may admit more than one distinct causal linear order; this happens, for example, in the case of two spacelike separated sets. When the members of a causally orderable set $\mathcal{K}$ are the coupling zones for a collection of observers, we will describe the observers as causally orderable and use any causal linear ordering of $\mathcal{K}$ to induce a linear order on the collection of observers.

Now let $\mathcal{S}$ be a theory of interest and consider two causally orderable observers, $\sf A$ and $\sf B$, with probe theories $\mathcal{P}_{\sf A}$ and $\mathcal{P}_{\sf B}$. The description of $\sf A$'s measurements in the FV scheme involves \emph{inter alia} the uncoupled combination $\mathcal{S}\otimes \mathcal{P}_{\sf A}$ and a coupled theory $\mathcal{C}_{\sf A}$ with a coupling zone $K_{\sf A}$, along with a corresponding scattering map $\Theta_{\sf A}$ on $\mathcal{S}\otimes \mathcal{P}_{\sf A}$; $\sf B$'s measurements are described in a similar way. If both $\sf A$ and $\sf B$ measure independently, then they can be considered as a combined ``super-observer'' whose probe theory $\mathcal{P}_{\sf \{A,B\}}$ is a tensor product of $\mathcal{P}_{\sf A}$ and
$\mathcal{P}_{\sf B}$. As the two coupling regions $K_{\sf A}$ and $K_{\sf B}$
may be separated by a Cauchy surface, it is reasonable to assume that there is a combined coupled theory $\mathcal{C}_{\sf \{A,B\}}$ with coupling zone $K_{\sf A}\cup K_{\sf B}$, and a scattering map $\Theta_{\sf \{A,B\}}$, which can be decomposed as an appropriate composition of the individual scattering maps.
Accordingly, we say that the combination of $\sf A$ and $\sf B$ respects \emph{bipartite causal factorisation}, if and only if the coupled theory $\mathcal{C}_{\sf \{A,B\}}$ exists and
\begin{equation}
    \begin{aligned}
    \Theta_{\sf \{A,B\}} = \begin{cases}
    		\hat{\Theta}_{\sf A} \circ \hat{\Theta}_{\sf B} \quad \text{ if } K_{\sf B} \cap J^-(K_{\sf A})= \emptyset\\
    		\hat{\Theta}_{\sf B} \circ \hat{\Theta}_{\sf A} \quad \text{ if } K_{\sf A} \cap J^-(K_{\sf B})= \emptyset,
    		\end{cases}
    \end{aligned}
\end{equation}
where $\hat{\Theta}_{\sf X}$ ($\sf X=A,B$) denotes the trivial extension of the scattering map $\Theta_{\sf X}$ from an automorphism of $\mathcal{S} \otimes \mathcal{P}_{\sf X}$ to an automorphism of $\mathcal{S} \otimes \mathcal{P}_{\sf \{A,B\}}$ by tensoring with a suitable identity map. In particular, if the two coupling regions are spacelike separated then 
$ \Theta_{\sf \{A,B\}}$ may be factored in both ways. The assumption of bipartite causal factorisation is motivated by the expression for the scattering map in terms of time-ordered products in conventional perturbation theory. As a special case of Bogoliubov's factorisation relation, bipartite causal factorisation is in particular fulfilled by the time-ordered exponential of local coupling terms (with smooth cutoff) in renormalised perturbation theory~\cite{rejzner2016} and by the generators in recent non-perturbative Lagrangian approaches~\cite{Buchholz2020}. Moreover it is proved to hold for the probe models considered in~\cite{fewster2018quantum}.

Our treatment of multiple observers is based on three physically motivated assumptions:
  \begin{enumerate}
        \item every finite collection of causally orderable observers can be combined to form a super-observer, whose probe theory is a tensor product of the individual probe theories;
        \item the combination process may be achieved in a single step, or equivalently, as the result of successive stages of combination; 
        \item the combination of any causally orderable pair of observers fulfils \emph{bipartite causal factorisation}.
    \end{enumerate}
    
    To illustrate these ideas, let us consider three (distinct) causally orderable observers $\sf A,B,C$ admitting a causal linear order $\leq$ in which $\sf A \leq B \leq C$. The super-observer $\sf \{A,B,C\}$ can be formed in one go, or equivalently by first combining $\sf A$ and $\sf B$ to $\sf \{A,B\}$ and then further combining with $\sf C$; alternatively, we can first combine $\sf B$ with $\sf C$ and then combine with $\sf A$. Understanding `equivalence' as equality of scattering maps, we have 
    \begin{equation}
        \Theta_{\sf \{\{A,B\},C\}}=  \Theta_{\sf \{A,B,C\}}  
        = \Theta_{\sf \{A,\{B,C\}\}} 
    \end{equation}
    and, on using bipartite causal factorisation, one has
    \begin{equation}
    \begin{aligned}
    \Theta_{\sf \{A,B,C\}}  = 
    \hat{\Theta}_{\sf \{A,B\}} \circ \hat{\Theta}_{\sf C} = 
    \hat{\Theta}_{\sf A} \circ \hat{\Theta}_{\sf B} \circ \hat{\Theta}_{\sf C},
    \end{aligned}
    \end{equation}
    where the hats denote the extension of the scattering maps to $\mathcal{S} \otimes\mathcal{P}_{\sf \{A,B,C\}}$. 
    Moreover the assumptions also imply that whenever there is a choice between different causal orders for fixed $\sf A,B,C$, the combined scattering map $\Theta_{\sf \{A,B,C\}}$ can be written as a composition of the individual scattering maps in either of these orders.
    
    In general, given any finite set $\sf Obs$ of $N$ causally orderable observers, the super-observer has a combined probe theory
    \begin{equation}
        \mathcal{P}_{\sf Obs} = \bigotimes_{\sf X\in Obs} \mathcal{P}_{\sf X}
    \end{equation}
    and an overall scattering map $\Theta_{\sf Obs}$ on $\mathcal{S}\otimes\mathcal{P}_{\sf Obs}$ that factorises as
    \begin{equation}
        \Theta_{\sf Obs} = \hat{\Theta}_{\sf X_1}\circ \hat{\Theta}_{\sf X_2}
        \circ \cdots \circ \hat{\Theta}_{{\sf X}_N}
    \end{equation}
    whenever ${\sf X}_1 < {\sf X}_2 < \ldots < {\sf X}_N$ according to some causal linear ordering $\le$ of $\sf Obs$; the hats denote extensions to $\mathcal{S}\otimes\mathcal{P}_{\sf Obs}$. There are many equivalent formulae for $\Theta_{\sf Obs}$, arising from different ways of successively combining the observers.

\section{Measurements by multiple observers}
\label{sec:multiple_obs}

In this section we demonstrate how multiple successive measurements can be treated in the FV framework. We start with a discussion of one single observer and a pair of two observers, where we recall results from~\cite{fewster2018quantum}. We then move on to present the treatment of three observers, which readily generalises to the general $N \in \mathbb{N}$ observer case. We end this section with a discussion of the process of post-selection.

\subsection{Induced observables and effects}

Let $\mathcal{S}$ be a theory of interest and let $\sf A$ be an observer who wishes to measure (the expectation value) of some local observable of $\mathcal{S}$ in initial state $\omega$. Suppose $\sf A$ has probe theory $\mathcal{P}_{\sf A}$, initial probe state $\sigma_{\sf A}$, compact coupling zone $K_{\sf A}$, a coupled theory $\mathcal{C}_{\sf A}$, identification maps and associated scattering map $\Theta_{\sf A}: \mathcal{S} \otimes \mathcal{P}_{\sf A} \to \mathcal{S} \otimes \mathcal{P}_{\sf A}$. The prediction for the expectation value of a probe observable $O_{\sf A} \in \mathcal{P}_{\sf A}$ in the actual experiment conducted by $\sf A$, given initial system state $\omega$, is denoted by $\expt_{\sf A}(O_{\sf A}; \omega)$ and given by~\cite{fewster2018quantum}
\begin{equation}
    \begin{aligned}
    \expt_{\sf A}(O_{\sf A}; \omega)=(\omega \otimes \sigma_{\sf A}) \qty(\Theta_{\sf A} (\openone \otimes O_{\sf A})).
    \end{aligned}
    \label{eq_one_obs_expect_value}
\end{equation}

If the observable algebras are represented as operators on a Hilbert space, we can consider the case where $\omega$ and $\sigma_{\sf A}$ are given by density matrices and where $\Theta_{\sf A}$ is implemented as the adjoint action of a unitary scattering operator. Then the above equation is easily recognised as a straightforward generalisation of Eq.~\eqref{eq_non_technical_single_exp_value}.

Returning to the general situation, it was shown in~\cite{fewster2018quantum} that $\sigma_{\sf A}$ and $\Theta_{\sf A}$ give rise to a map $\varepsilon_{\sf A}: \mathcal{P}_{\sf A} \to \mathcal{S}$ such that
\begin{equation}
    \begin{aligned}
    \forall O_{\sf A} \in \mathcal{P}_{\sf A}: \; \omega\qty(\varepsilon_{\sf A} (O_{\sf A}))=(\omega \otimes \sigma_{\sf A}) \qty(\Theta_{\sf A} (\openone \otimes O_{\sf A})).
    \end{aligned}
    \label{eq_induced_obs}
\end{equation}
We call $\varepsilon_{\sf A} (O_{\sf A})$ the \emph{induced (system) observable} corresponding to $O_{\sf A}$, as introduced in Eq.~\eqref{eq_non-technical_induced_observable}.

Exploiting the \emph{Haag property} it can be shown that $\varepsilon_{\sf A} (O_{\sf A})$ is a local observable of the system theory, which can be localised in any connected region containing $K_{\sf A}$~\cite[Theorem~3.3]{fewster2018quantum}.

The interpretation of Eq.~\eqref{eq_induced_obs} is the following: if observer $\sf A$ is interested in the expectation value of a specific local system observable, then she needs to prepare and tune her physical detector, i.e., find $O_{\sf A}$, $\sigma_{\sf A}$ and $\Theta_{\sf A}$ such that $\varepsilon_{\sf A}(O_{\sf A})$ is the desired system observable. It is an open question whether this is always possible, so we take the pragmatic viewpoint and say the system observables of interest to an observer are those which can be measured using a probe, i.e., those which can be induced by some probe observable upon tuning the probe state and scattering map.

The result of an actual experiment is generically not immediately a sharp numerical outcome but rather an answer to a (finite collection of) yes-no question(s). In quantum theory, this is modelled abstractly by considering a projector or more generally an \emph{effect} $P$ associated to \emph{yes} and $\openone - P$ associated to \emph{no} as the main observables of interest. Recall that a $*$-algebra element $P$ is an effect if and only if $P^\dagger =P$ and $0 \leq P \leq \openone$ where $ 0 \leq P$ means that $P$ is a convex combination of elements of the form $A^\dagger A$. One frequently calls the expectation value of an effect $P$ its \emph{success probability}.

\subsection{Two observers}

Consider a set of two observers $\sf Obs= \{ A, B\}$ each of whom wishes to determine the expectation value of a system observable $\varepsilon_{\sf A}(O_{\sf A})$ and $\varepsilon_{\sf B}(O_{\sf B})$, respectively, for probe observables $O_{\sf A} \in \mathcal{P}_{\sf A}$ and  $O_{\sf B} \in \mathcal{P}_{\sf B}$. We intend to answer the following question: \emph{``What is the expected outcome of observer $\sf B$'s measurement?''}

Similar to before, for every ${\sf X} \in {\sf Obs}$ who interacts with a system-theory $\mathcal{S}$ in initial state $\omega$ we have a probe-theory $\mathcal{P}_{\sf X}$, initial state $\sigma_{\sf X}$, compact coupling zone $K_{\sf X}$, coupled theory $\mathcal{C}_{\sf X}$ identification maps and associated scattering map $\Theta_{\sf X}: \mathcal{S} \otimes \mathcal{P}_{\sf X} \to \mathcal{S} \otimes \mathcal{P}_{\sf X}$.

As in the previous section, we may regard the collection of all observers as super-observer in its own right with probe $\mathcal{P}:= \bigotimes_{{\sf X} \in {\sf Obs}} \mathcal{P}_{\sf X}$, coupling zone $K:= \bigcup_{{\sf X} \in {\sf Obs}} K_{\sf X}$ and initial state $\sigma:= \bigotimes_{{\sf X} \in {\sf Obs}} \sigma_{\sf X} $ on $\mathcal{P}$. Let $O:= \bigotimes_{{\sf X} \in {\sf Obs}} O_{\sf X}$ be the probe observable of interest. We assume the existence of an associated coupled theory $\mathcal{C}$ emerging from coupling $\mathcal{S}$ to $\mathcal{P}$ in $K$ giving rise to a scattering map $\Theta:\mathcal{S} \otimes \mathcal{P} \to \mathcal{S} \otimes \mathcal{P}$.

Let us now assume that after (an ensemble of) their experiments, the observers meet in their joint future to analyse their experimental data together. Since we consider the two of them to constitute a single super-observer, the expectation value of their joint measurement is, according to Eq.~\eqref{eq_one_obs_expect_value},
\begin{equation}
    \begin{aligned}
\expt_{\sf \{A,B\}}(O;\omega)=   (\omega \otimes \sigma) \qty(\Theta(\openone \otimes O)).
    \end{aligned}
    \label{eq_two_obs_one_super_probe_A}
\end{equation}
If $O_{\sf A}$ and $O_{\sf B}$ are effects, $O$ is an effect as well and Eq.~\eqref{eq_two_obs_one_super_probe_A} can be understood as the success probability of the ``combined effect'' $O$ corresponding to the success of both $O_ {\sf A}$ and $O_{ \sf B}$, i.e., $\expt_{\sf \{A,B\}}(O;\omega)  =\Prob_{\sf \{A,B\}}(O_{\sf A}\& O_{\sf B};\omega)$.

In the context of effects it is also immediately possible to give an answer to the posed question. $\sf B$'s expected outcome is the success probability $ \Prob_{\sf \{A,B\}}(O_{\sf B};\omega)$ that $\sf B$ observes probe effect $O_{\sf B}$ given initial system state $\omega$ irrespective of $\sf A$'s outcome. It is given as a marginal probability
\begin{equation}
    \begin{aligned}
    &\Prob_{\sf \{A,B\}}(O_{\sf B};\omega)\\
    &= \Prob_{\sf \{A,B\}}(O_{\sf A}\& O_{\sf B};\omega) + \Prob_{\sf \{A,B\}}((\neg O_{\sf A})\& O_{\sf B};\omega)\\
    &=  (\omega \otimes \sigma) \qty(\Theta(\openone \otimes O_{\sf A} \otimes O_{\sf B})) \\
    &\quad + (\omega \otimes \sigma) \qty(\Theta(\openone \otimes (\openone - O_{\sf A}) \otimes O_{\sf B}))\\
    &=(\omega \otimes \sigma) (\Theta(\openone \otimes \hat{O}_{\sf B})),
    \end{aligned}
    \label{eq_two_obs_motivation}
\end{equation}
where we used an explicit order of the tensor product, $\mathcal{P}= \mathcal{P}_{\sf A} \otimes \mathcal{P}_{\sf B} $ and where the hat denotes the inclusion of $O_{\sf B}$ in $\mathcal{P}$, $\hat{O}_{\sf B}=\openone\otimes O_{\sf B}$. (The case in which one wishes to perform an analysis post-selected on a specific outcome of observer $\sf A$ is known as \emph{selective} measurement and is discussed further at the end of this section.)

Note that Eq.~\eqref{eq_two_obs_motivation} only depends on $\sf A$ through $\sigma_{\sf A}$ and the coupled scattering map $\Theta$ and, in particular, is independent of $O_{\sf A}$. As a matter of fact, the above discussion in terms of effects can be seen as a motivation for considering Eq.~\eqref{eq_two_obs_motivation} to be the answer to the posed question even in the situation where $O_{\sf A}$ and $O_{\sf B}$ are \emph{not} effects, i.e., generally:
\begin{equation}
    \begin{aligned}
    \expt_{\sf \{A.B\} }(O_{\sf B};\omega)=(\omega \otimes \sigma) (\Theta(\openone \otimes \hat{O}_{\sf B})).
    \end{aligned}
    \label{eq_two_obs_answer}
\end{equation}
This is useful since, e.g., the field $*$-algebra of the linear scalar field does not admit any non-trivial effects at all. Moreover, the expression makes a prediction for the \emph{expectation value} of $\sf B$'s experiment, which can be determined from $\sf B$'s local experimental data alone. There is no need for $\sf B$ to meet $\sf A$ in their joint future to conduct data analysis together.

Let us continue the investigation of expression~\eqref{eq_two_obs_answer} in the physically relevant case that the set $\sf Obs$ is \emph{causally orderable}. There are at most two possible linear causal orders on $\sf Obs$, corresponding to the cases 1. $\sf A \leq B$; or 2. $\sf B \leq A$.

If the combination of the two observers respects bipartite causal factorisation, the super-scattering map decomposes as $\Theta = \hat{\Theta}_{\sf A} \circ \hat{\Theta}_{\sf B}$ in the first case, while $\Theta = \hat{\Theta}_{\sf B} \circ \hat{\Theta}_{\sf A}$ in the second. Before continuing, it is convenient to observe that, upon writing $\hat{\Theta}_{\sf B}(\openone \otimes O_{\sf A} \otimes O_{\sf B}) = \sum_j S_j \otimes O_{\sf A} \otimes B_j$ and noting that $\varepsilon_{\sf B}(O_{\sf B}) = \sum_j \sigma_{\sf B}(B_j) S_j$, the following holds:
\begin{equation}
    \begin{aligned}
   &(\omega \otimes \sigma_{\sf A} \otimes \sigma_{\sf B})(\hat{\Theta}_{\sf A} \circ \hat{\Theta}_{\sf B}(\openone \otimes O_{\sf A} \otimes O_{\sf B}))\\
   &=(\omega \otimes \sigma_{\sf A} \otimes \sigma_{\sf B})(\hat{\Theta}_{\sf A}( \sum_j S_j \otimes O_{\sf A} \otimes B_j))\\
   &=(\omega \otimes \sigma_{\sf A})(\Theta_{\sf A}( \sum_j \sigma_{\sf B}(B_j) S_j \otimes O_{\sf A}))\\
   &= (\omega \otimes \sigma_{\sf A})(\Theta_{\sf A} (\varepsilon_{\sf B}(O_{\sf B}) \otimes O_{\sf A})).
   \end{aligned}
   \label{eq_observation}
    \end{equation}
This allows us to simplify Eq.~\eqref{eq_two_obs_answer} in each case:\\

1. For $\sf A \leq B$ we order the tensor product of probes as $\mathcal{P}=\mathcal{P}_{\sf A}\otimes\mathcal{P}_{\sf B}$ and get
    \begin{equation}
    \begin{aligned}
   \expt_{\sf \{A,B\}}(O_{\sf B};\omega)&= (\omega \otimes \sigma)(\Theta(\openone \otimes \hat{O}_{\sf B}))\\
   &= (\omega \otimes \sigma_{\sf A} \otimes \sigma_{\sf B})(\hat{\Theta}_{\sf A} \circ \hat{\Theta}_{\sf B}(\openone \otimes \hat{O}_{\sf B}))\\
   &= (\omega \otimes \sigma_{\sf A})(\Theta_{\sf A} (\varepsilon_{\sf B}(O_{\sf B}) \otimes \openone)).
   \end{aligned}
   \end{equation}
    Therefore, if the system state $\omega_{\sf A}$ is defined so that $\omega_{\sf A}(C):=(\omega\otimes\sigma_{\sf A})(\Theta_{\sf A}(C\otimes \openone))$ for all $C\in\mathcal{S}$, $\sf B$'s expected outcome in this situation is  
   \begin{equation}
       \expt_{\sf \{A,B\}}(O_{\sf B};\omega) =\omega_{\sf A}(\varepsilon_{\sf B}(O_{\sf B}))=\expt_{\sf B}(O_{\sf B};\omega_{\sf A}),
   \label{eq_two_obs_ordered_ans}
    \end{equation}
    which is his expected outcome if $\sf A$ does not measure, but with the system prepared in state $\omega_{\sf A}$ instead of $\omega$. This is the justification for regarding $\omega_{\sf A}$ as the updated system state associated with $\sf A$'s measurement, as asserted in Eqs.~\eqref{eq_qm_single_state_update} and~\eqref{eq_fv_single_state_update}. For future reference, let us define the update map $\mathcal{J}_{\sf A}(\omega):= \omega_{\sf A}$.
    
    2. For $\sf B \leq A$ and after ordering the tensor product of probes as $\mathcal{P}=\mathcal{P}_{\sf B}\otimes\mathcal{P}_{\sf A}$ for convenience,
    \begin{equation}
    \begin{aligned}
    \expt_{\sf \{A,B\}}(O_{\sf B};\omega)&=(\omega \otimes \sigma)(\Theta(\openone \otimes \hat{O}_{\sf B}))\\
   &= (\omega \otimes \sigma_{\sf B} \otimes \sigma_{\sf A})(\hat{\Theta}_{\sf B} \circ \hat{\Theta}_{\sf A}(\openone \otimes O_{\sf B}))\\
   &= (\omega \otimes \sigma_{\sf B})(\Theta_{\sf B} (\openone \otimes O_{\sf B}))\\
   &= \omega(\varepsilon_{\sf B}(O_{\sf B}))=\expt_{\sf B}(O_{\sf B};\omega),
   \end{aligned}
    \end{equation}
    where we used that $\hat{\Theta}_{\sf A}(\openone \otimes \hat{O}_{\sf B}) =(\openone \otimes \hat{O}_{\sf B})$. This follows from the fact that in the present order of the tensor product of probes we have that $\openone \otimes \hat{O}_{\sf B}= \openone  \otimes O_{\sf B}\otimes \openone$ and $\hat{\Theta}_{\sf A} = \Theta_{\sf A} \otimes_2 \openone$, where again the subscript on the tensor product indicates the slot into which the second factor is inserted. Recalling that observer $\sf B$  precedes $\sf A$ (with respect to $\leq$), the above result shows that there is no influence from the future to the past.

Finally, we remark that if $K_{\sf A}$ and $K_{\sf B}$ are spacelike separated, the causal order is not unique: there is an ordering corresponding to case 1 and another corresponding to case 2. However, there is no ambiguity because $\mathsf{B}$'s expected outcome is given by~\eqref{eq_two_obs_answer} in either case. This implies in particular that
\begin{equation}
\omega_{\sf A}(\varepsilon_{\sf B}(O_{\sf B})) =\omega(\varepsilon_{\sf B}(O_{\sf B}))
\end{equation}
if $K_{\sf A}, K_{\sf B}$ are spacelike separated (as has been observed in~\cite{fewster2018quantum}).

\subsection{Three observers}

The obvious next step is to consider three observers and to give an answer to the question: \emph{``For a set of three observers $\sf Obs = \{A,B,C\}$, each of which performs a measurement, what is the expected outcome of observer $\sf B$'s measurement of the induced system observable $\varepsilon_{\sf B}(O_{\sf B})$?''}

Following the reasoning and notation of before we can immediately write down the answer as
\begin{equation}
    \begin{aligned}
 \expt_{\sf \{A,B,C\}}(O_{\sf B};\omega) =  (\omega \otimes \sigma) (\Theta(\openone \otimes \hat{O}_{\sf B})),
    \end{aligned}
     \label{eq_three_obs_answer}
\end{equation}
where we again assumed the existence of an appropriate overall super-observer similar to the bipartite case.

Let us further investigate Eq.~\eqref{eq_three_obs_answer} under the additional assumption that $\sf Obs$ is causally orderable and that causal factorisation holds. As $\sf A$ and $\sf C$ can be interchanged, there are at most three cases: $\sf A,B,C$ are such that there exists a linear order $\leq$ with (1) $\sf A \leq B \leq C$; (2) $\sf C \leq A \leq B$; or (3) $\sf B \leq C \leq A$.

Choosing a convenient order of tensor products and using results from before yields the following:\\

1. For $\sf A \leq B \leq C$:
    \begin{equation}
    \begin{aligned}
   &\expt_{\sf \{A,B,C\}}(O_{\sf B};\omega)\\
   &= (\omega \otimes \sigma_{\sf A} \otimes \sigma_{\sf B} \otimes \sigma_{\sf C})(\hat{\Theta}_{\sf A} \circ \hat{\Theta}_{\sf B} \circ \hat{\Theta}_{\sf C} (\openone \otimes \hat{O}_{\sf B}))\\
    &= (\omega \otimes \sigma_{\sf A} \otimes \sigma_{\sf B})(\check{\Theta}_{\sf A} \circ \check{\Theta}_{\sf B}(\openone \otimes \check{O}_{\sf B}))\\
   &= (\omega \otimes \sigma_{\sf A})(\Theta_{\sf A} (\varepsilon_{\sf B}(O_{\sf B}) \otimes \openone))\\
   &= \omega_{\sf A}(\varepsilon_{\sf B}(O_{\sf B}))\\
   &=\expt_{\sf B}(O_{\sf B};\omega_{\sf A}),
   \end{aligned}
   \label{eq_three_obs_ordered_ABC}
    \end{equation}
    where the ha\v{c}ek denotes the extension to $\mathcal{S} \otimes \mathcal{P}_{\{ \sf A,B\}}$ and where, similarly to before, we used that $\hat{\Theta}_{\sf C} (\openone \otimes \hat{O}_{\sf B}) = \openone \otimes \hat{O}_{\sf B}$ and that $\hat{\Theta}_{\sf A}$ and $\hat{\Theta}_{\sf B}$ act trivially on $\mathcal{P}_{\sf C}$. The upshot is that observer $\sf B$'s outcome is given by taking the initial system state, updating it according to the map $\mathcal{J}_{\sf A}$ associated to the observer preceding $\sf B$ (with respect to $\leq$) and evaluating the updated state on $\sf B$'s induced system observable. The observer succeeding $\sf B$ (with respect to $\leq$) can be completely ignored.\\
    
    2. For $\sf C \leq A \leq B$, ordering  $\mathcal{P}=\mathcal{P}_{\sf C}\otimes\mathcal{P}_{\sf A}\otimes\mathcal{P}_{\sf B}$,
    \begin{equation}
    \begin{aligned}
   &\expt_{\sf \{A,B,C\}}(O_{\sf B};\omega)\\
   &= (\omega \otimes \sigma_{\sf C} \otimes \sigma_{\sf A} \otimes \sigma_{\sf B})(\hat{\Theta}_{\sf C} \circ \hat{\Theta}_{\sf A} \circ \hat{\Theta}_{\sf B} (\openone \otimes \hat{O}_{\sf B}))\\
    &= (\omega \otimes \sigma_{\sf C} \otimes \sigma_{\sf A})(\check{\Theta}_{\sf C} \circ \check{\Theta}_{\sf A}( \varepsilon_{\sf B}(O_{\sf B} \otimes \openone))\\
   &= \omega_{\sf AC}(\varepsilon_{\sf B}(O_{\sf B}))\\
   &=\expt_{\sf B}(O_{\sf B};\omega_{\sf AC}),
   \end{aligned}
   \label{eq_three_obs_ordered_ACB}
    \end{equation}
    where $\omega_{\sf AC} := \qty(\mathcal{J}_{\sf A} \circ \mathcal{J}_{\sf C}) (\omega)$, cf.~Theorem~3.5 in~\cite{fewster2018quantum}. Here, the ha\v{c}ek denotes the trivial extension of the scattering maps to $\mathcal{S} \otimes \mathcal{P}_{\{\sf A, C\}}$.
The investigation of this case has some interesting consequences. First it provides a proof of Eq.~\eqref{eq_successive_update} which we used in the discussion of Sorkin's protocol (where Alice takes the place of $\sf C$ and Bob takes the place of $\sf A$ here). Second, if we regard $\sf X:= \{\sf A,C\}$ as a super-observer in its own right, then Eq.~\eqref{eq_three_obs_ordered_ACB} can be written as 
\begin{equation}
    \begin{aligned}
&\expt_{\sf \{A,B,C\}}(O_{\sf B};\omega)\\
   &= (\omega \otimes \sigma_{\sf X} \otimes \sigma_{\sf B})(\hat{\Theta}_{\sf X} \circ \hat{\Theta}_{\sf B}(\openone \otimes \hat{O}_{\sf B}))\\
   &= (\omega \otimes \sigma_{\sf X})(\Theta_{\sf X} (\varepsilon_{\sf B}(O_{\sf B}) \otimes \openone))\\
   &= \omega_{\sf X}(\varepsilon_{\sf B}(O_{\sf B}))\\
   &=\expt_{\sf B}(O_{\sf B};\omega_{\sf X}),
   \end{aligned}
    \end{equation}
    which is the same calculation as in the case of two observers leading to Eq.~\eqref{eq_two_obs_ordered_ans}. This idea will be used in the remaining case, as well as later on to simplify the investigation of $N > 3$ observers. (We will continue to use capital sans-serif Latin letters for individual observers, i.e., elements of $\sf Obs$, \emph{as well as} super-observers, i.e., subsets of $\sf Obs$.)\\
    
    3. For $\sf B \leq C \leq A$, ordering  $\mathcal{P}=\mathcal{P}_{\sf B}\otimes\mathcal{P}_{\sf C}\otimes\mathcal{P}_{\sf A}$ and then regarding $\sf X:= \{A,C\}$ as super-probe in its own right enables us to write
    \begin{equation}
    \begin{aligned}
    &\expt_{\sf \{A,B,C\}}(O_{\sf B};\omega)\\
   &= (\omega \otimes \sigma_{\sf B} \otimes \sigma_{\sf X})(\hat{\Theta}_{\sf B} \circ \hat{\Theta}_{\sf X}  (\openone \otimes \hat{O}_{\sf B}))\\
   &= (\omega \otimes \sigma_{\sf B})(\Theta_{\sf B} (\openone \otimes \hat{O}_{\sf B}))\\
   &= \omega(\varepsilon_{\sf B}(O_{\sf B}))\\
   &=\expt_{\sf B}(O_{\sf B};\omega),
   \end{aligned}
    \end{equation}
    which reinforces the message that there is no signalling from the future to the past.
    
    A given set of observers $\sf \{A,B,C\}$ can admit more than one causal order, however, as for two observers, the answer for the various admissible cases will agree.
    
\subsection{$N$ observers}
Let us assume that we have a finite set $\sf Obs$ of $N$ causally orderable observers, each of whom wishes to determine the expectation value of an induced system observable. We fix one observer $\sf B \in Obs$ and ask this: \emph{``What is the expected outcome $\expt_{\sf Obs}(O_{\sf B};\omega)$ of observer $\sf B$'s measurement of the induced system observable $\varepsilon_{\sf B}(O_{\sf B})$?''} More specifically, we want to know how the general answer $\expt_{\sf Obs}(O_{\sf B};\omega)= (\omega \otimes \sigma) (\Theta(\openone \otimes \hat{O}_{\sf B}))$ may be simplified in this situation.

Any fixed causal order $\leq$ on $\sf Obs$ gives rise to the following tripartite partition $\sf A:= \{ X \in Obs| X< B \}$, $\{\sf B\}$ and $\sf C:=\{ X \in Obs| X> B \}$. It immediately follows from Eq.~\eqref{eq_three_obs_ordered_ABC} in the analysis of three observers that

\begin{equation}
    \begin{aligned}
    &\expt_{\sf Obs}(O_{\sf B};\omega)= \omega_{\sf A}(\varepsilon_{\sf B}(O_{\sf B}))=\expt_{\sf B}(O_{\sf B};\omega_{\sf A}).
   \end{aligned}
   \label{eq_N_obs_ans_}
    \end{equation}
    That is, $\sf B$'s expected outcome is equal to the one obtained in the absence of the other observers, but in  system state $\omega_{\sf A}$ (which is equal to $\omega$ if $\sf A$ is empty, i.e., if $\sf B$ is the earliest observer according to $\le$). 
    Moreover, suppose that the constituent observers of $\sf A$
are ordered $\sf A_1<A_2<\cdots< A_{|A|}$ according to $\le$. Then causal factorisation gives
    \begin{equation}
    \begin{aligned}
   \Theta_{\sf A} = \hat{\Theta}_{\sf A_1} \circ \cdots \circ  \hat{\Theta}_{{\sf A_{|A|}}},
   \end{aligned}
\end{equation}
which implies that we can write the updated state $\omega_{\sf A}$ as the result of $\sf |A|$ many individual updates according to
\begin{equation}
    \begin{aligned}
   \omega_{\sf A}&=(\mathcal{J}_{\sf A_{|A|}} \circ \cdots \circ\mathcal{J}_{\sf A_1}) (\omega)=\omega_{\sf A_{|A|}, ... , {\sf A_1}},
   \end{aligned}
\end{equation}
which follows from the ${\sf|A|}=2$ case by induction.

This demonstrates how multiple measurements are modelled in the FV framework and shows in particular how the familiar concept of successive state-updates is recovered in the situation of causally orderable observers. As we have emphasised, in spite of the possible ambiguity of the causal order, \emph{causal factorisation} ensures that the answer is unique and also free of any influence from the future to the past with respect to any causal order on the set of observers. However, Sorkin's impossible measurements raise the question of whether any rule for assigning expectation values might be plagued by other acausal influences. It is the purpose of Sec.~\ref{sec:Sorkin-N} to show that this is not the case in the FV framework, thus generalising the results of Sec.~\ref{sec:main_result}.

Before that, we end the present section with a discussion of \emph{selective} measurements, i.e., the process of post-selection and associated conditional expectation values.

\subsection{Conditional expectation and post-selection}

Let $\sf Obs=\{ A, B\}$ be a set of two causally orderable observers each of whom performs a measurement of system observable $\varepsilon_{\sf X}(O_{\sf X})$ for $\sf X \in Obs$. In this context: \emph{``What is observer $\sf B$'s outcome conditioned on a certain outcome of $\sf A$'s measurement?''} 

To this end let us consider the situation where $O_{\sf A}$ is an effect and where $\sf B$ is interested in the outcome of his experiment given that $\sf A$ measures successfully. As before, let us motivate our answer by first considering the case of an \emph{effect} $O_{\sf B}$ and let us assume that, after conducting (an ensemble of) their experiments, the observers meet in their joint future to analyse their data together. In this case (slightly generalising Sec.~3.3 of~\cite{fewster2018quantum}) we can write down the success probability of $O_{\sf B}$ conditioned on (success of) $O_{\sf A}$ as the conditional probability
\begin{equation}
    \begin{aligned}
\Prob_{\sf \{A,B\}}(O_{\sf B} | O_{\sf A};\omega)&=   \frac{\Prob_{\sf \{A,B\}}(O_{\sf B} \& O_{\sf A};\omega)}{\Prob_{\sf \{A,B\}}( O_{\sf A};\omega)}\\
&=   \frac{(\omega \otimes \sigma) \qty(\Theta(\openone \otimes   O_{\sf A}\otimes O_{\sf B}))}{(\omega \otimes \sigma) (\Theta(\openone \otimes \hat{O}_{\sf A}))}\\
&=\frac{\omega(\varepsilon_{\sf \{A, B\}}(  O_{\sf A}\otimes O_{\sf B}))}{\omega(\varepsilon_{\sf A}(O_{\sf A}))},
    \end{aligned}
    \label{eq_two_obs_selective_prob}
\end{equation}
under the standing assumption that $\Prob_{\sf \{A,B\}}(O_{\sf A};\omega)=(\omega \otimes \sigma) (\Theta(\openone \otimes \hat{O}_{\sf A})) \neq 0$, and where we used the explicit order of tensor products $\mathcal{P} = \mathcal{P}_{\sf A} \otimes \mathcal{P}_{\sf B}$. We emphasise that the conditional success probability is operationally determined by means of \emph{post-selection,} i.e., the \emph{selection} of those members of the ensemble of the combined experiment, which yielded a positive answer to $\sf A$'s measurement. This requires access to the experimental data of both $\sf B$ and $\sf A$ and can consequently only be performed in their joint future.

Having found the conditional success probability, we again view it as a justification for postulating the following conditional expectation as an answer to the question in the case where $O_{\sf B}$ is not necessarily an effect but a general observable:
\begin{equation}
    \begin{aligned}
\expt_{\sf \{A,B\}}(O_{\sf B} | O_{\sf A};\omega)&=\frac{\omega(\varepsilon_{\sf \{A, B\}}(  O_{\sf A}\otimes O_{\sf B}))}{\omega(\varepsilon_{\sf A}(O_{\sf A}))}.
    \end{aligned}
    \label{eq_two_obs_selective}
\end{equation}

Let us investigate Eq.~\eqref{eq_two_obs_selective} further in the following cases:

1. $\sf A \leq B:$ Using Eq.~\eqref{eq_observation} we observe that
    \begin{equation}
    \begin{aligned}
    \omega(\varepsilon_{\sf \{A, B\}}(O)) = (\omega \otimes \sigma_{\sf A})(\Theta_{\sf A}(\varepsilon_{\sf B}(O_{\sf B}) \otimes O_{\sf A}),
    \end{aligned}
    \end{equation}
    which allows us to write
    \begin{equation}
    \begin{aligned}
    \expt_{\sf \{A,B\}}(O_{\sf B} | O_{\sf A};\omega)&=   \frac{(\omega \otimes \sigma_{\sf A})(\Theta_{\sf A}(\varepsilon_{\sf B}(O_{\sf B}) \otimes O_{\sf A})}{(\omega \otimes \sigma) (\Theta(\openone \otimes \hat{O}_{\sf A}))}\\
    &=\frac{(\omega \otimes \sigma_{\sf A})(\Theta_{\sf A}(\varepsilon_{\sf B}(O_{\sf B}) \otimes O_{\sf A})}{(\omega \otimes \sigma_{\sf A}) (\Theta_{\sf A}(\openone \otimes O_{\sf A}))}\\
    &=: \omega_{{\sf A} | O_{\sf A}}(\varepsilon_{\sf B}(O_{\sf B})) \\ 
    &= \expt_{\sf B}(O_{\sf B};\omega_{{\sf A} | O_{\sf A}}),
    \label{eq_selective_two_obs}
    \end{aligned}
    \end{equation}
    in terms of the \emph{selective} update map \begin{equation}
    \mathcal{J}_{{\sf A} | O_{\sf A}}(\omega):= \omega_{{\sf A} | O_{\sf A}},
    \end{equation}
    which yields a well-defined state provided that $O_{\sf A}$ has nonzero success probability $(\omega \otimes \sigma_{\sf A}) (\Theta_{\sf A}(\openone \otimes O_{\sf A}))\neq 0$~\cite{fewster2018quantum}. Equation~\eqref{eq_selective_two_obs} shows how Eq.~\eqref{eq_two_obs_selective} can be understood in terms of an updated state in the case where there exists a causal order such that $\sf A$ precedes $\sf B$ and also constitutes a proof of Eq.~\eqref{eq_non_technical_conditional_expectation}. Additionally note that $\mathcal{J}_{{\sf A} | \openone}=\mathcal{J}_{\sf A}$.\\
    
    2. $\sf B \leq A:$ the interpretation of this scenario is that $\sf B$ performs post-selection on a measurement that (at least with respect to one causal order) \emph{succeeds} his own. There is \emph{a priori} no reason to expect this post-selection to be trivial and we have not found any simplified expression for $\expt_{\sf \{A,B\}}(O_{\sf B} | O_{\sf A};\omega)$ in this case. One might naively think that such a post-selection conflicts causality as there is an \emph{apparent} influence from the future to the past. This issue is resolved by reminding oneself that post-selection can only be performed by all observers together in their joint future.

For completeness we mention that it was shown in Theorem~3.4 in~\cite{fewster2018quantum} how in the case of spacelike separation, any possible \emph{apparent} acausal behaviour of the selective update map can be attributed to spacelike correlations of the initial state $\omega$. To see this we observe that for spacelike separated $\sf A, B$ 
\begin{equation}
\begin{aligned}
&\hat{\Theta}_{\sf A} \circ \hat{\Theta}_{\sf B} (\openone \otimes O_{\sf A} \otimes O_{\sf B}) \\
&=\hat{\Theta}_{\sf A} \circ \hat{\Theta}_{\sf B} (\openone \otimes O_{\sf A} \otimes \openone \cdot  \openone \otimes \openone \otimes O_{\sf B})\\
&= \qty(\hat{\Theta}_{\sf A} \circ \hat{\Theta}_{\sf B} (\openone \otimes O_{\sf A} \otimes \openone)) \cdot \qty(  \hat{\Theta}_{\sf A} \circ \hat{\Theta}_{\sf B} (\openone \otimes \openone \otimes O_{\sf B}))\\
&= \qty(\hat{\Theta}_{\sf A} \circ \hat{\Theta}_{\sf B} (\openone \otimes O_{\sf A} \otimes \openone)) \cdot \qty(  \hat{\Theta}_{\sf B} \circ \hat{\Theta}_{\sf A} (\openone \otimes \openone \otimes O_{\sf B}))\\
&= \qty( \Theta_{\sf A} (\openone \otimes O_{\sf A}) \otimes \openone) \cdot \qty(  \Theta_{\sf B}(\openone \otimes O_{\sf B}) \otimes_2 \openone),
\end{aligned}
\end{equation}
which (cf. Sec. 3.2 in~\cite{fewster2018quantum}) implies that
\begin{equation}
    \begin{aligned}
        \omega(\varepsilon_{\sf \{A, B\}}(  O_{\sf A}\otimes O_{\sf B})) = \omega(\varepsilon_{\sf A}(  O_{\sf A}) \varepsilon_{\sf B}( O_{\sf B})).
    \end{aligned}
\end{equation}
This shows that for spacelike separated observers
\begin{equation}
    \begin{aligned}
\expt_{\sf \{A,B\}}(O_{\sf B} | O_{\sf A};\omega)=\expt_{\sf B}(O_{\sf B};\omega)
    \end{aligned}
\end{equation}
if and only if $\omega(\varepsilon_{\sf A}(O_{\sf A}) \varepsilon_{\sf B}(O_{\sf B})) = \omega(\varepsilon_{\sf A}(O_{\sf A})) \omega(\varepsilon_{\sf B}(O_{\sf B}))$, i.e., $\varepsilon_{\sf A}(O_{\sf A})$ and $\varepsilon_{\sf B}(O_{\sf B})$ are uncorrelated in state $\omega$.

The generalisation to $N$ observers in the case where $\sf Obs$ admits a causal order $\leq$ such that $\sf B$ is the \emph{latest} observer (with respect to $\leq$) follows immediately: we look at the partition $\sf A:= \{ X \in Obs| X < B\}$ and $\{\sf B\}$ and wish to condition $\sf B$'s expected outcome of a measurement of the probe-effect $O_{\sf B}$ on the successful measurement of the probe-effects $O_{\sf X}$ for $\sf X \in A $. Setting $O_{\sf A}:= \bigotimes_{\sf X \in A} O_{\sf X}$ yields just as before
\begin{equation}
\begin{aligned}
\expt_{\sf Obs}(O_{\sf B} | O_{\sf A};\omega)= \omega_{{\sf A} | O_{\sf A}}(\varepsilon_{\sf B}(O_{\sf B})) = \expt_{\sf B}(O_{\sf B};\omega_{{\sf A} | O_{\sf A}}).
\end{aligned}
\end{equation}
It is noteworthy that the total selective update map $\mathcal{J}_{ {\sf A} | O_{\sf A}}(\omega)$ can be written as a composition of individual selective update maps 
\begin{equation}
\begin{aligned}
\mathcal{J}_{ {\sf A} | O_{\sf A}}(\omega)=(\mathcal{J}_{\sf |A|} \circ \cdots \circ \mathcal{J}_1)(\omega),
\end{aligned}
\end{equation}
where we again ordered the constituent observers of $\sf A$ as $\sf A_1 < \cdots < A_{|A|}$ and used the short-hand notation $\mathcal{J}_r:=\mathcal{J}_{ {\sf A}_{r} | O_{{\sf A}_{r}}}$. This follows by induction from the case $|\mathsf{A}|=2$ given by Theorem~3.5
in~\cite{fewster2018quantum}.

\section{Absence of impossible measurements for multiple observers}
\label{sec:Sorkin-N}

In this section we demonstrate the absence of any acausal influence in the measurements of an arbitrary finite number of causally orderable observers in a theory respecting causal factorisation. 

To that end let us reconsider the situation of $N$ observers $\sf Obs$. As in the previous section, we focus our attention on a fixed observer $\sf B$, taking the role played by Charlie in Sec.~\ref{sec:main_result}, and a fixed linear order $\leq$. As shown before in Eq.~\eqref{eq_N_obs_ans_}, $\sf B$'s expected outcome $\expt_{\sf Obs}(O_{\sf B};\omega)$ equals $\omega_{\sf A}(\varepsilon_{\sf B}(O_{\sf B}))$ in the absence of any post-selection on results of any other observers. Let us assume that there is an observer $\sf Y \in A$ who is spacelike separated from $\sf B$, i.e, $K_{\sf B} \subseteq K^\perp_{{\sf Y}}$, and will play the role of Alice. This gives rise to the partition $\sf A=X\cup \{Y\} \cup Z$, where $\sf X:= \{ J \in A| J< Y \}$ and $\sf Z:=\{ J  \in Obs| J> Y \}$. The super-observer $\sf Z$ will play the role of Bob. We can then write
\begin{equation}
    \begin{aligned}
   \omega_{\sf A}(\varepsilon_{\sf B}(O_{\sf B})) = (\omega_{\sf X} \otimes \sigma_{{\sf Y}} \otimes \sigma_{\sf Z})(\hat{\Theta}_{{\sf Y}} \circ \hat{\Theta}_{\sf Z}(\varepsilon_{\sf B}(O_{\sf B}) \otimes \openone)).
   \end{aligned}
\end{equation}
The following holds:
\begin{theo}
If, in the above notation, $K_{\sf B}$ is connected and spacelike separated from $K_{{\sf Y}}$, then:
\begin{equation}
    \begin{aligned}
   &(\omega_{\sf X} \otimes \sigma_{{\sf Y}} \otimes \sigma_{\sf Z})(\hat{\Theta}_{{\sf Y}} \circ \hat{\Theta}_{\sf Z}(\varepsilon_{\sf B}(O_{\sf B}) \otimes \text{\rm $\openone$}))\\
   &=(\omega_{\sf X} \otimes \sigma_{\sf Z})(\Theta_\mathsf{Z}(\varepsilon_{\sf B}(O_{\sf B}) \otimes \text{\rm $\openone$})).
   \end{aligned}
\end{equation}
\label{theo_general_no_signalling}
\end{theo}
A consequence of this theorem is that
\begin{equation}
    \begin{aligned}
   \omega_{\sf A}(\varepsilon_{\sf B}(O_{\sf B})) = \omega_{\sf A\setminus \{ Y \}}(\varepsilon_{\sf B}(O_{\sf B})),
   \end{aligned}
\end{equation}
and hence 
\begin{equation}
    \expt_{\sf Obs}(O_{\sf B};\omega) = \expt_{\sf A\setminus\{Y\}}(O_{\sf B};\omega) ,
\end{equation}
emphasising that we can completely ignore the spacelike separated observer ${\sf Y}$ as well as the super-observer $\sf C$ succeeding $\sf B$ with respect to $\leq$.

It follows by successive application of the theorem that no observer that is spacelike separated from $\sf B$ can influence the expected outcome of observer $\sf B$'s measurement. This shows that there is no Sorkin-type (or any other) superluminal signalling between the individuals in the $N$ observer case if each coupling zone is connected. We remark that the need to restrict to connected coupling zones comes from the connectedness condition in the formulation of the Haag property. If connectedness is dropped from the Haag property (cf. footnote~\cite{Note1}), then one can also drop the connectedness in the above theorem.\\

The proof of Theorem \ref{theo_general_no_signalling} relies on the geometrical Lemma \ref{lem_topological} and an application of Theorem \ref{theo_no_signalling}.

\begin{lem}
For every connected compact subset $K \subseteq M$ there exists a connected region $N \supseteq K$ with compact closure.
\label{lem_topological}
\end{lem}

\begin{proof}
A subset with compact closure is called \emph{precompact}. As a smooth manifold, $M$ has an exhaustion by countably many precompact open sets $G_\alpha$ such that $\overline{G}_\alpha \subseteq G_{\alpha +1}$ and $M= \bigcup_\alpha G_\alpha$. Since $K$ is compact, it can be covered by finitely many $G_\alpha$ and since they are nested, there exists $\beta$ such that $K \subseteq G_\beta$. Since $M$ is globally hyperbolic, the causal hull of $G_\beta$, $\ch(G_\beta)$, is open (see Lemma A.8 in~\cite{Fewster_2012}) and $\ch(\overline{G}_\beta)$ is compact (by definition of global hyperbolicity). Since obviously $G_\beta \subseteq \overline{G}_\beta$, it also follows that $\ch(G_\beta) \subseteq \ch(\overline{G}_\beta)$ [as $J^\pm(G_\beta) \subseteq J^\pm(\overline{G}_\beta)$ and $\ch(G_\beta) = J^-(G_\beta) \cap J^+(G_\beta)$]. So $\ch(G_\beta)$ is a subset of a compact set and hence precompact. As $\ch(G_\beta)$ can be viewed as a globally hyperbolic manifold in its own right, every connected component is hence precompact, connected, open and causally convex. Since $K$ is connected and contained in $\ch(G_\beta)$, it is contained in one connected component, which finishes the proof.
\end{proof}
Let us now prove Theorem \ref{theo_general_no_signalling}.
\begin{proof}[Proof of Theorem \ref{theo_general_no_signalling}]
By assumption, $K_{\sf B} \subseteq K_{{\sf Y}}^\perp$ and from the existence of the causal order $K_{\sf B} \subseteq M^+_{\sf Z}= M \setminus J^-(K_{\sf Z})$, and since $K_{\sf B}$ is connected, it is contained in one connected component of the open, causally convex subset $K_{{\sf Y}}^\perp \cap M_{\sf Z}^+$. This connected component can be viewed as a globally hyperbolic manifold in its own right and hence we can apply Lemma \ref{lem_topological} to furnish a connected region $N \subseteq K_{{\sf Y}}^\perp \cap M_{\sf Z}^+$ that contains $K_{\sf B}$ and has compact closure fully contained in $K_{{\sf Y}}^\perp$. Now $\varepsilon_{\sf B}(O_{\sf B})$ is localisable in $N$, moreover, after identifying $K_1 := K_{{\sf Y}}$, $K_2:= K_{\sf Z}$ and $O_3:=N$, we see that the assumptions of Theorem \ref{theo_no_signalling} are fulfilled, thus establishing the desired equality.
\end{proof}

\section{Conclusions and outlook}
\label{sec:conclusion}

The issue of measurement in QFT has been plagued by acausality exemplified by Sorkin's protocol. Our main result shows that a consistent and fully causal interpretation of tripartite measurement processes in the sense of measurement schemes is possible via the local and covariant proposal in~\cite{fewster2018quantum}, which is applicable to generic quantum field theories coupled to external forces and on possibly curved spacetimes. The principle of \emph{causal factorisation} of scattering processes for an arbitrary finite number of causally orderable observers allowed us to generalise our result to the $N$ observer case. As opposed to other work, such as~\cite{borsten2019impossible}, our result thereby provides a class of ``physically allowed operations'' that can be characterised abstractly \emph{as well as} constructed explicitly in specific models, see~\cite{fewster2018quantum}.

The FV framework may be considered a first important step towards a solution to the problem of delineating \emph{all} `physically allowed quantum operations' raised in~\cite{beckman2002}; however, whether all of them are induced by FV measurement schemes is unknown. It is therefore important to more explicitly characterize the system observables associated to measurement schemes. We intend to report on this issue elsewhere. It is also worth noting that local scattering operators, understood as operations reflecting the result of measuring observables, have recently been proposed as a new foundation for AQFT~\cite{Buchholz2020} and this viewpoint could be fruitfully combined with ours. 

For our key assumption of causal factorization, we have restricted ourselves to a physical motivation and formulation of this assumption. A more rigorous, mathematical investigation would be very interesting and we intend to report on this elsewhere.

Finally it is worth mentioning that \emph{non-relativistic}, \emph{non-local} particle detector models are a very common tool and widely used for example in quantum field theory in curved spacetime and relativistic quantum information. Other authors have shown that coupling such a detector model to a finite number of field modes~\cite{benincasa2014quantum} or to all but the zero mode~\cite{tjoa2019zero} leads to superluminal signalling. In view of our result, this is due to the non-locality of such a coupling, whereas the detector model sketched at the end of~\cite{fewster2018quantum} consequently does not signal superluminally. 
Applying the FV framework to questions in which particle detector models have so far been used, for example \emph{entanglement harvesting}~\cite{Pozas-Kerstjens_2015}, promises to yield additional insight both on a conceptual level and with respect to explicitly computable results.

\begin{acknowledgments}

We thank Roger Colbeck for useful comments on the text and Ian Jubb and Leron Borsten for stimulating discussions on the general topic of impossible measurements. The work of M.~H.~R.~was supported by a Mathematics Excellence Programme Studentship awarded by the Department of Mathematics at the University of York.

\end{acknowledgments}

%\bibliography{apssamp}% Produces the bibliography via BibTeX.

\providecommand{\noopsort}[1]{}\providecommand{\singleletter}[1]{#1}%

\end{document}